\documentclass[english,a4paper,12pt]{article}
\pdfoutput=1

\usepackage[utf8]{inputenc}
\usepackage[T1]{fontenc}
\usepackage{titling}

\usepackage{amsmath, amssymb}
\usepackage{pgf}
\usepackage{amsfonts}
\usepackage{dsfont}
\usepackage{mathrsfs}
\usepackage{mathabx}
\usepackage{stmaryrd}
\usepackage{makeidx}
\usepackage{amsbsy}
\usepackage{amsthm}
\usepackage{color}
\usepackage{fullpage}
\usepackage{enumitem}
\usepackage[vcentermath]{youngtab}
\usepackage{marginnote} 
\usepackage{mdframed}	
\newmdenv[topline=false, bottomline=false, skipabove=\topsep, skipbelow=\topsep]{siderules}
\newtheorem{theorem}{Theorem}

\newtheorem{proposition}{Proposition}
\newtheorem{lemma}{Lemma}

\newtheorem{definition}{Definition}

\def\sD{\mathscr{D}}

\def\sH{\mathscr{H}}

\def\sK{\mathscr{K}}

\def\sP{\mathscr{P}}

\def\bS{{\mathbb S}}

\def\cA{{\mathfrak A}}
\def\cB{{\mathfrak B}}

\def\cM{{\mathfrak M}}

\newcommand{\ca}[1]{{\cal #1}}
\newcommand{\ben}{\begin{equation}}
\newcommand{\een}{\end{equation}}
\def\bena{\begin{eqnarray}}
\def\eena{\end{eqnarray}}

\def\cA{{\ca A}}
\def\cB{{\ca B}}

\def\cM{{\ca M}}

\renewcommand{\H}{\mathscr{H}}

\def\1{{\mathds{1}}}

\newcommand{\dd}{{\rm d}}
\newcommand{\tr}{\operatorname{Tr}}

\renewcommand{\log}{\operatorname{ln}}

\renewcommand{\epsilon}{\varepsilon}

\newcommand{\RR}{\mathbb{R}}
\newcommand{\CC}{\mathbb{C}}

\renewcommand{\Re}{{\rm Re}}

\begin{document}
\title{Approximate recoverability and relative entropy II: 
2-positive channels of general v. Neumann algebras
}

	\author{Thomas Faulkner$^{1}$\thanks{\tt\tt tomf@illinois.edu}, Stefan Hollands$^{2}$\thanks{\tt stefan.hollands@uni-leipzig.de}\\
	{\it $^1$ University of Illinois at Urbana-Champaign, IL}\\
	{\it $^2$ ITP, Universit\" at Leipzig, MPI-MiS Leipzig and KITP, Santa Barbara}
	}

\date{\today}
	
\maketitle

\begin{abstract}
We generalize our results in paper I in this series to quantum channels between general v. Neumann algebras, 
proving the approximate recoverability of states which undergo a small change in 
relative entropy through the channel. 
To this end, we derive a strengthened form of the quantum data processing inequality for the change in 
relative entropy of two states under a channel between two v. Neumann algebras.
Compared to the usual inequality, there is an explicit lower bound involving the fidelity 
between the original state and a 
recovery channel.
\end{abstract}

\section{Introduction}

The relative entropy between two density operators $\rho, \sigma$, defined as 
\ben\label{eq:Srel}
S(\rho| \sigma) = \tr[\rho (\log \rho - \log \sigma)], 
\een
is an asymptotic measure of their distinguishability. Classically, $e^{-NS(\{p_i\} | \{q_i\})}$ approaches for large $N$ 
the probability for a sample of size $N$ of letters, distributed according to the 
true distribution $\{p_i\}$, when calculated according to an incorrect guess $\{q_i\}$. In the non-commutative setting, the relative entropy was later generalized to the 
general v. Neumann algebras by Araki \cite{Araki1,Araki2} using relative modular hamiltonians. 

By far the most fundamental property of the relative entropy -- from which in fact essentially all others follow -- is its monotonicity under a channel. A channel is roughly an arbitrary combination of (i) a unitary time evolution of the density matrix, (ii) a v. Neumann measurement followed by post-selection, (iii) forgetting part of the system (partial trace). The fundamental property is that if $T$ is such a channel and its application to a density matrix is $\rho \circ T$
(understood in the sense of expectation functional, $\rho \circ T (a)=\tr(\rho T(a))$, then always
\ben\label{eq:dpi}
S(\rho | \sigma) \ge S( \rho \circ T| \sigma\circ T). 
\een
In quantum information theory, $T$ is related to data processing, so \eqref{eq:dpi} is sometimes called the data-processing inequality (DPI).
It implies for instance the strong subadditivity property of the v. Neumann entropy \cite{Lieb}. 
After important special cases were proven by \cite{Lindbladt,Araki1}, the data-processing inequality was demonstrated in the setting of general v. Neumann algebras
by Uhlmann \cite{Uhlmann1977}.

Suppose the distinguishability decreases only by a small amount $\epsilon$ under $T$. Then one would like to say that, given $\sigma, T$, 
the state $\rho$ can be recovered ``with high fidelity'' from $\rho\circ T$. For this, one should establish an improved a 
lower bound on the quantity $S(\rho | \sigma) - S( \rho \circ T | \sigma \circ T)$ in terms of the fidelity. 
This issue has been considered by several authors such as \cite{Junge, Berta, Wilde, Carlen, Jencova, Sutter, Sutter2}, 
partly based on earlier characterizations of the case of equality in \eqref{eq:dpi} \cite{Petz1,Petz2,Petz3,JencovaPetz}. 

A particularly attractive lower bound 
in the DPI has recently been given by Junge et al. \cite{Junge} for v. Neumann algebras of type I. 
They consider a certain recovery channel $\alpha$ that is closely related to the 
Petz-map (see e.g. \cite{Petz1993}, prop. 8.3 and sec. \ref{sec:4} below) which is constructed in a canonical way 
out of $T$ and $\sigma$, but does not depend on $\rho$. They then show that 
there is a lower bound in the DPI of the form of an averaged log-fidelity,  
\ben
\label{eq:dpiimp}
S(\rho | \sigma) - S( \rho \circ T| \sigma \circ T) \ge 
-\log 
F(\rho  | \rho \circ T \circ \alpha)^2 ,
\een
where $F$ is the ``square root'' fidelity between two states \eqref{eq:Fchar}.

It is interesting to note that, while \cite{Junge} prove it only for type I v. Neumann algebras\footnote{I.e., direct sums of matrix algebras or the algebra of all bounded operators on a separable Hilbert space.}, the ingredients of this inequality formally make sense for channels between arbitrary v. Neumann algebras. Thus, it is a natural question whether the inequality generalizes to that setting. We are motivated in particular by recent applications of the improved form of the DPI in high energy physics \cite{swingle,Faulkner}, where such questions are natural in the algebraic approach to quantum field theory \cite{haag_2}. In this context, the v. Neumann algebras under consideration are of type III$_1$ \cite{Buchholz}. On the other hand, the proof by 
\cite{Junge} strongly uses the special properties of type I v. Neumann algebras. In this paper, we give a proof of \eqref{eq:dpiimp} for arbitrary sigma-finite v. Neumann algebras, see thm. \ref{thm:2}.  While we were finishing this manuscript \cite{Junge:2020bwo} appeared which reported results that have some overlap with ours, albeit with different methods. 

Our result generalizes paper I where the case of an embedding $T$ of arbitrary v. Neumann algebras was considered -- corresponding to a ``partial trace'' of density operators in the type I context. The proof strategy is broadly similar to paper I but some new features arise which we deal with in this paper. In section 
\ref{sec:2} we recall how the ingredients in our inequality \eqref{eq:dpiimp} are defined in the setting of general v. Neumann algebras. Then in section \ref{sec:3} we state the main results indicating 
also some generalizations. In section \ref{sec:4} we make basic remarks on a connection with the Jones index of v. Neumann subfactors, and in section 
\ref{sec:5} we give the proof of our main theorems. 

\medskip
\noindent
{\bf Notations and conventions:} Calligraphic letters $\cA, \cM, \dots$ denote v. Neumann algebras. Calligraphic letters $\sH, \sK, \dots$ denote linear spaces.  $\bS_a=\{z \in \CC \mid 0 < \Re (z) < a\}$ denotes a strip. We use the physicist's ``ket''-notation $|\psi\rangle$ for vectors in a Hilbert space. 
The scalar product is written as 
\ben
(|\psi\rangle, |\psi'\rangle)_{\sH} =: \langle \psi | \psi' \rangle
\een
and is anti-linear in the first entry. The norm of a vector is written simply as
$\| |\psi\rangle \| =: \| \psi \|$. The action of a linear operator $T$ on a ket is sometimes written as $T|\phi\rangle = |T\phi\rangle$. 
In this spirit, the norm of a bounded linear operator $T$ on $\sH$ is written as $\|T\|= \sup_{|\psi\rangle: \|\psi\|=1} \|T\psi\|$. 

\section{Relative entropy, data processing inequality and Petz map}\label{sec:2}

\subsection{Relative entropy}

Let $(\cM, J, \sP_\cM^\natural, \sH)$ be a v. Neumann algebra in standard form acting on a Hilbert space $\sH$, 
with natural cone $\sP^\sharp_\cM$ and modular conjugation $J$ (see paper I \cite{Faulkner:2020iou} for our notations and \cite{Bratteli,Takesaki} as general references). 
As in paper I  \cite{Faulkner:2020iou}, we use relative modular operators $\Delta_{\eta,\psi}$ associated with two vectors 
$|\eta\rangle, |\psi\rangle \in \sH$ in our constructions. 
According to \cite{Araki1,Araki2}, if $\pi^\cM(\eta) \ge \pi^\cM(\psi)$, the relative entropy may be defined in terms of them by 
\ben
S(\psi | \eta) = -\lim_{\alpha \to 0^+} \frac{\langle \xi_\psi | \Delta^\alpha_{\eta, \psi} \xi_\psi \rangle-1}{\alpha} ,  
\een
otherwise, it is by definition infinite. Here, $|\xi_\psi\rangle$ denotes the unique representer of 
a vector $|\psi\rangle$ in the natural cone. The relative entropy only depends on the 
functionals $\omega_\psi, \omega_\eta$ on $\cM$, but not the choice of vectors $|\psi\rangle, |\eta\rangle$ that define these functionals. 
We will therefore use interchangeably the notations $S(\psi | \eta)=S(\omega_\psi | \omega_\eta)$.
The definition may also be rewritten as follows.  
For $t \in \RR$,  the Connes-cocycle is the isometric operator from $\cM$ defined by 
\ben
(D\psi : D\eta)_{t} = \Delta^{it}_{\psi,\eta} \Delta^{-it}_{\eta,\eta} 
\een
which can also be extracted from $ \Delta^{it}_{\psi,\psi} \Delta^{-it}_{\eta,\psi} = (D\psi : D\eta)_{t}  \pi^{\mathcal{M}'}(\psi) $.
As the relative modular operator, this only depends on the state functionals $\omega_\psi, \omega_\eta$.
In terms of the Connes-cocycle, the relative entropy may also be defined as
\ben
S(\psi | \eta) = -i\frac{\dd}{\dd t} \langle \xi_\psi | (D\eta : D\psi)_{t} \xi_\psi \rangle |_{t=0} = -i\frac{\dd}{\dd t} \omega_\psi( (D\eta : D\psi)_{t} ) |_{t=0}.  
\een
The last expression has the advantage that it does not require one to know the vector representative 
of $|\psi \rangle$ in the natural cone. The derivative exists whenever $S(\psi | \eta)<\infty$ \cite{Petz1993}, thm. 5.7.
In the case of the matrix algebra $M_n(\CC)$, where $\omega_\eta$ and $\omega_\psi$ are identified with density matrices, 
the relative entropy is the usual expression \eqref{eq:Srel}.

\subsection{Data processing inequality}

The basic situation studied in this paper is the following. $\cB, \cA$ are two v. Neumann algebras, assumed to be in standard form. They act on Hilbert spaces $\sK, \sH$, 
with corresponding natural cones $\sP^\sharp_\cA, \sP^\sharp_\cB$ and associated anti-linear unitary maps $J_\cA, J_\cB$ (so that $J_\cA |\xi\rangle = |\xi\rangle$ for $|\xi\rangle \in \sP^\sharp_\cA$, and similarly for $\cB$) . 
A {\em channel} is a normal, i.e. ultra weakly continuous, linear mapping 
\ben
T: \cB \to \cA
\een
such that $T(1)=1$ and such that any non-negative self-adjoint element of $\cB$ gets mapped to to such an element of $\cA$, i.e. 
$T(b^*b) \ge 0$.
$T$ is not required to be a homomorphism (but could be). We will use the following standard terminology. 

\begin{definition}
\begin{enumerate}
\item 
A channel $T$ is called a Schwarz map if for any $a \in \cB$ we have Kadison's property
\ben\label{eq:eqschwarz}
T(a^*)T(a) \le T(a^*a).
\een
\item
A channel $T$ is called 2-positive if any non-negative element from 
the $2\times 2$ matrix algebra $\cB \otimes M_2(\CC)$ with entries in $\cB$ 
gets mapped under $T \otimes 1_{M_2}$ to a non-negative element of $\cA \otimes M_2(\CC)$.
\end{enumerate}
\end{definition}

A 2-positive map is also a Schwarz map, see e.g. \cite{Petz4}, thm. E. In particular, both properties follow if $T=\iota$ is 
an embedding of v. Neumann algebras, i.e. if $\cB$ is a v. Neumann subalgebra of $\cA$ under a $*$-homomorphism $\iota$.
That case was treated in paper I  \cite{Faulkner:2020iou}. 

If $T$ is a linear, normal mapping $T: \cB \to \cA$, then we define its adjoint ${\cA}_* \to {\cB}_*$, operating between 
the corresponding spaces of normal linear functionals by duality.
Let us assume that we have two normal state functionals $\omega_\eta, \omega_\psi$ on $\cA$ induced by vectors from $\sH$. 
Then the pull backs $ \omega_\eta \circ T, \omega_\psi \circ T$ give normal state functionals on $\cB$. We shall write 
\ben\label{eq:conventions}
\omega_\eta(a) \equiv \langle \eta_\cA | a \eta_\cA \rangle, \quad 
\omega_\eta \circ T(b) \equiv \langle \eta_\cB | b \eta_\cB \rangle, 
\een
with suitable vector representatives $|\eta_\cA\rangle \in \sH, |\eta_\cB\rangle \in \sK$ that we may choose. The (unique) representers  in the natural cones
will be denoted by $|\xi_{\eta}^\cA \rangle \in \sH, |\xi_{\eta}^\cB \rangle \in \sK$.
A similar notation is adopted when $|\eta\rangle$ is replaced by  $|\psi \rangle$. 

The quantum data processing inequality (DPI) \cite{Uhlmann1977} states that 
\ben
S(\omega_\psi | \omega_\eta) \ge S(\omega_\psi \circ T | \omega_\eta \circ T), 
\een
viewing the relative entropy as
a function on state functionals,
or equivalently that $S(\psi_\cA | \eta_\cA) \ge S(\psi_\cB | \eta_\cB)$ viewing it as a function on vectors. 
One of our main results, thm. \ref{thm:1}, will be an improved, explicit, lower bound on this inequality
in the spirit of \cite{Junge}, generalizing our earlier paper I where the special case of inclusions $T$ was treated. 

\subsection{Petz recovery map}
\label{subsec:petz}

We now recall the definition of the Petz map in the case of general v. Neumann algebras, discussed in more detail in \cite{Petz1993}, sec. 8.
Let $T: \cB \to \cA$ be a unital, normal, and 2-positive (or Schwarz) map between two v. Neumann algebras. Let $\omega_\eta$ be a normal 
state functional on $\cA$ with pull-back $ \omega_\eta \circ T$ to $\cB$. If $\omega_\eta$ is not faithful for $\cA$, then as in paper I \cite{Faulkner:2020iou}
we consider instead the v. Neumann subalgebra $\cA_\pi = \pi^\cA(\eta) \cA \pi^\cA(\eta) \pi^{\cA'}(\eta)$, noting that this does not change 
the relative entropy $S(\psi | \eta)$ since by definition $\pi^\cA(\psi) \le \pi^\cA(\eta)$ unless the entropy is infinite (in which case all of our 
inequalities are trivial). Likewise, letting $\pi^\cB(\eta)$ be the support projection of $\omega_\eta \circ T$ in $\cB$, if $\omega_\eta \circ T$
is not faithful, we pass to $\cB_\pi = \pi^\cB(\eta) \cB \pi^\cB(\eta)\pi^{\cB'}(\eta)$. Finally, we pass from $T:\cB \to \cA$ to $T_\pi( \, . \, ) = \pi^\cA(\eta) T( \, . \, )
\pi^\cA(\eta)\pi^{\cA'}(\eta)$ which is again a normal, unital 2-positive (or Schwarz) map $T_\pi:  \cB_\pi \to \cA_\pi$. 
By such kinds of constructions, we can and will restrict ourselves from now on to states such that both $ \omega_\eta \circ T, \omega_\eta$ are faithful. Furthermore, by basic properties of the GNS construction, we may assume that $|\eta_\cB\rangle, |\eta_\cA\rangle$ are cyclic and separating. See paper I for details, which are similar here. 

First we recall the definition of the Petz-map. Let $|\eta\rangle$ be cyclic and separating for a v. Neumann algebra $\cM$. Then the KMS scalar product 
on $\cM$ regarded as a vector space is defined as 
\ben
\langle m,n \rangle_\eta = \langle \eta | m^* \Delta_\eta^{1/2} n \eta \rangle . 
\een
Since we have such vectors $|\eta_\cA\rangle$ and $|\eta_\cB\rangle$ for both $\cA$ and $\cB$, we can use the KMS scalar products to 
define the adjoint $T^+:\cA \to \cB$ (depending on the choices of these vectors) of a normal, unital and 2-positive 
$T: \cB \to \cA$ by results of Petz, see \cite{Petz1993} prop. 8.3, 
who shows that $T^+$ is well-defined, normal, unital, and 2-positive. The rotated Petz map, which we call $\alpha_{\eta,T}^t: \cA \to \cB$, 
is defined by conjugating this with the respective modular flows, i.e. 
\ben
\alpha_{\eta,T}^t = \varsigma_{\eta,\cB}^t \circ T^+ \circ \varsigma_{\eta,\cA}^{-t}
\een
where $\varsigma_{\eta,\cA}^t = {\rm Ad} \Delta_{\eta, \cA}^{it}$ is the modular flow for $\cA, |\eta_\cA\rangle$ etc. An equivalent definition is:

\begin{definition}
For $b \in \cB, a \in \cA$, and $T$ is unital, normal, and 2-positive, the rotated Petz map $\alpha_{\eta,T}^t: \cA \to \cB$ is defined implicitly by the identity:
\ben
\label{eq:Petzdef}
\langle b \xi_{\eta}^{\cB} | J_\cB \Delta_{\eta_\cB}^{it} \alpha_{\eta,T}^t(a) \xi_{\eta}^{\cB} \rangle = 
\langle T(b) \xi_\eta | J_\cA \Delta_{\eta_\cA}^{it} a \xi_{\eta}^{\cA} \rangle. 
\een
\end{definition}

The following is then a trivial consequence of \cite{Petz1993}, prop. 8.3:

\begin{lemma}
The map $\alpha_{\eta,T}^t: \cA \to \cB$ is well-defined, normal, unital, and 2-positive for all $t \in \RR$.
\end{lemma}

For the case of finite dimensional type I factors, the definition of the rotated Petz map is easily seen to coincide with that given by \cite{Junge}.

\section{Main theorems}\label{sec:3}

We now state our main theorems. 
Let $T: \cB \to \cA$ be a normal, unital Schwarz map and let $\omega_\psi, \omega_\eta$
be normal state functionals on $\cA$.
We introduce following \cite{Petz4,Petz2} a linear map $V_{\psi}: \sK \to \sH$ by
\ben
\label{eq:Vdef}
V_{\psi}b|\xi_{\psi}^{\cB}\rangle := T(b) |\xi_{\psi}^{\cA} \rangle \quad (b \in \cB).
\een
As it stands, the definition is actually consistent only when $|\xi_{\psi}^{\cB}\rangle$ is cyclic and separating. 
In the general case, one can define \cite{Petz4} instead
\ben
\label{eq:Vdef1}
V_{\psi}(b |\xi_{\psi}^{\cB}\rangle + |\Omega\rangle) := T(b) |\xi_{\psi}^{\cB} \rangle \quad (b \in \cB,
\pi^{\cB'}(\psi) |\Omega\rangle =0).
\een
In either case, it easily follows from Kadison's property \eqref{eq:eqschwarz} that $V_{\psi}$ is a contraction $\|V_{\psi}\| \le 1$, see 
e.g. \cite{Petz4}, proof of thm. 4. 

The first main result is a strengthened 
version of the DPI. To this end, we introduce a vector valued function 
\ben\label{eq:Gupper}
t \mapsto |\Gamma_\psi(1/2+it)\rangle := 
\Delta_{\eta_\cA,\psi_\cA}^{1/2 + it} V_{\psi} \Delta_{\eta_\cB,\psi_\cB}^{-1/2 -it} |\xi_{\psi}^{\cB}\rangle \quad (t\in \RR), 
\een
the existence and properties of which are established in lem. \ref{lem:1} and lem. \ref{lem:5} below. In particular, the representation \eqref{eq:Gupper1}
shows in conjunction with Stone's theorem (see e.g. \cite{specth}, sec. 5.3) that this function is strongly continuous.

\begin{theorem}\label{thm:1}
If $T: \cB \to \cA$ is a normal, unital Schwarz map and $q \in [1,2]$, we have
\ben
\label{eq:main}
S(\omega_\psi | \omega_\eta) - S(\omega_\psi \circ T | \omega_\eta \circ T)
 \ge 
- \int_{-\infty}^\infty p(t) \log \| \Gamma_\psi(1/2+it)\|_{q,\psi}^2 \, \dd t \ . 
\een
Here $\| \zeta \|_{q,\psi}$ denotes the  Araki-Masuda-H\" older $L_q(\cA', \psi)$-norm relative to $\psi,\cA'$, and 
$p(t) := \pi \left[ 1+\cosh(2\pi t) \right]^{-1}$, which is a probability density. 
\end{theorem}

The proof is given in sec. \ref{sec:5} and is similar to that of thm. 1, paper I  \cite{Faulkner:2020iou}. We will therefore only provide complete details where the proof deviates significantly. 

\medskip
\noindent
{\bf Remarks:}
1) The non-commutative H\" older spaces on a v. Neumann algebra $\cM$ weighted by a state $\phi$, called $L_p(\cM, \phi)$, 
were defined and analyzed in \cite{AM} (for $p \in [1,\infty]$). Their definition is recalled for convenience in app. B. 
A recent reference including some inequalities which we will use is \cite{Berta2} (see also \cite{JencovaLp1}).
A variational characterization of the $L_p$ norms has been given by \cite{Hollands2}.

2) In paper I, this result was established in the case when $T=\iota: \cB \subset \cA$ is an inclusion of v. Neumann algebras.

\medskip
\noindent
To state our second main theorem, we define 
\ben\label{eq:recov}
\alpha =   \int_{-\infty}^\infty p(t) \alpha_{\eta,T}^t \, \dd t, 
\een
where $\alpha_{\eta,T}^t$ is the rotated Petz map associated with $T$ and $|\eta\rangle$ as discussed in sec. \ref{subsec:petz}. This map
is a 2-positive channel -- called ``recovery channel'' -- provided that $T$ has these properties.
We also define the ``square root'' fidelity $F=F_\cA$ between two normal states on $\cA$ as usual by
\ben
 \label{eq:Fchar}. 
\| \zeta \|_{1,\psi,\cA'} = \sup\{ |\langle \zeta | a' \psi \rangle | : a' \in \cA', \|a'\|=1 \} = F_{\cA}(\omega_\zeta, \omega_\psi).
\een

\begin{theorem}\label{thm:2}
Let $T: \cB \to \cA$ be a unital, normal, and 2-positive map between two v. Neumann algebras. Let $\omega_\eta,\omega_\psi$ be normal 
state functionals on $\cA$ with $|\eta\rangle$ faithful. Then
\ben
S(\omega_\psi | \omega_\eta) - S(\omega_\psi \circ T | \omega_\eta \circ T)
 \ge -\log F(\omega_\psi | \omega_\psi \circ T \circ  \alpha)^2.
\een
\end{theorem}
\noindent
 {\bf Remarks:}
1) In the case of v. Neumann algebras of type I, our result reduces to that by Junge et al. \cite{Junge}. \\
2) The case of an inclusion $T=\iota :\cB \subset \cA$ was treated separately already in paper I.\\
3) By one of the Fuchs-van-der-Graff inequalities, see e.g. paper I  \cite{Faulkner:2020iou}, lem. 3 (3) for the case of general v. Neumann algebras, 
and the elementary bound $-\log x \ge 1-x, x \in (0,1]$, we obtain an upper bound on the norm distance between the original and recovered state:
\ben
S(\omega_\psi | \omega_\eta) - S(\omega_\psi \circ T | \omega_\eta \circ T)
 \ge \tfrac{1}{4} \| \omega_\psi \circ T \circ  \alpha - \omega_\psi\|^2.
\een
\medskip

\begin{proof}
Consider the linear operator $V_\eta: \sK \to \sH$ defined by 
\ben
V_\eta(\pi^\cB(\psi) b|\xi_{\eta}^{\cB} \rangle+ |\zeta\rangle) = \pi^\cA(\psi)T(b)|\xi_{\eta}^{\cB}\rangle.
\een 
Here, $b \in \cB, |\zeta\rangle \in [\pi^\cB(\psi) \cB \xi_{\eta}^{\cB}]^\perp$. 
It has been shown in \cite{Petz4}, thm. 4 that $\|V_\eta \| \le 1$, which uses the 2-positivity of $T$, 
and that $V_\eta S_{\eta_\cB, \psi_\cB} \subset S_{\eta_\cA, \psi_\cA} V_\psi$, where 
we mean the closures of the Tomita operators. We find:

\ben\label{eq:Gupper1}
\begin{split}
|\Gamma_\psi(1/2+it)\rangle
=&\Delta_{\eta_\cA,\psi_\cA}^{1/2 + it} V_{\psi} \Delta_{\eta_\cB,\psi_\cB}^{-1/2 -it} |\xi_{\psi}^{\cB}\rangle\\
=&\Delta_{\eta_\cA,\psi_\cA}^{it} J_\cA S_{\eta_\cA, \psi_\cA} V_{\psi} \Delta_{\eta_\cB,\psi_\cB}^{-1/2 -it} |\xi_{\psi}^{\cB}\rangle\\
=&\Delta_{\eta_\cA,\psi_\cA}^{it} J_\cA V_\eta S_{\eta_\cB, \psi_\cB} \Delta_{\eta_\cB,\psi_\cB}^{-1/2 -it} |\xi_{\psi}^{\cB}\rangle\\
=&\Delta_{\eta_\cA,\psi_\cA}^{it} J_\cA V_{\eta} J_{\cB} \Delta_{\eta_\cB,\psi_\cB}^{-it} |\xi_{\psi}^{\cB}\rangle. 
\end{split}
\een
Let $\gamma_t$ be the non-normalized linear functional on $\cA$ 
corresponding to $|\Gamma(1/2+it)\rangle$. We claim, compare paper I  \cite{Faulkner:2020iou}, thm. 4 (2):

\begin{lemma}
\label{lem:mon}
We have $\gamma_t(a^* a) \le \omega_\psi \circ T \circ \alpha^t_{\eta,T}(a^* a)$ for $t \in \RR, a \in \cA$.
\end{lemma}

\begin{proof}
In \eqref{eq:Petzdef}, set $b=B^* B, a=A^*A$. We get:
\ben
\label{eq:Petzdef1}
\langle  \Delta_{\eta_\cB}^{-it} J_\cB B \xi_{\eta}^{\cB} |  \alpha_{\eta,T}^t(A^*A) \Delta_{\eta_\cB}^{-it} J_\cB B \xi_{\eta}^{\cB} \rangle 
= \langle \Delta_{\eta_\cA}^{-it} J_\cA A \xi_{\eta}^{\cA} | T(B^* B) \Delta_{\eta_\cA}^{-it} J_\cA A \xi_{\eta}^{\cA} \rangle. 
\een
Now, for $b' \in \cB'$, and the state $\psi$ on $\cA$, we choose 
\ben
B:= J_\cB \Delta^{it}_{\eta_\cB} b' \Delta^{-it}_{\eta_\cB} J_\cB. 
\een
Then \eqref{eq:Petzdef1} becomes:
\ben
\begin{split}
&\langle  b' \xi_{\eta}^{\cB} |  \pi^\cB(\psi) \alpha_{\eta,T}^t(A^*A)  \pi^\cB(\psi) b'\xi_{\eta}^{\cB} \rangle \\
&\ge \langle \Delta_{\eta_\cA}^{-it} J_\cA A \xi_{\eta}^{\cA} | T(B^*)T(B) \Delta_{\eta_\cA}^{-it} J_\cA A \xi_{\eta}^{\cA} \rangle \\
&\ge \langle \Delta_{\eta_\cA}^{-it} J_\cA A \xi_\eta | T(B^*) \pi^\cA(\psi) T(B) \Delta_{\eta_\cA}^{-it} J_\cA A \xi_{\eta} \rangle \\
&= \langle \pi^\cA(\psi) T(B) \xi_{\eta}^{\cA} | \Delta_{\eta_\cA}^{-it} J_\cA A^* A J_\cA \Delta_{\eta_\cA}^{it} \pi^\cA(\psi) T(B) \xi_{\eta}^{\cA} \rangle \\
&= \langle V_\eta B \xi_{\eta}^{\cB} | J_\cA\Delta_{\eta_\cA,\psi_\cA}^{-it}  A^* A  \Delta_{\eta_\cA,\psi_\cA}^{it} J_\cA V_\eta B \xi_{\eta}^{\cB} \rangle \\
&= \langle \Delta_{\eta_\cA,\psi_\cA}^{it} J_\cA 
V_\eta J_\cB \Delta^{it}_{\eta_\cB,\psi_\cB}b' \xi_{\eta}^{\cB} |   (A^* A) \Delta_{\eta_\cA,\psi_\cA}^{it} J_\cA 
V_\eta J_\cB \Delta^{it}_{\eta_\cB,\psi_\cB}b' \xi_{\eta}^{\cB} \rangle.
\end{split}
\een
The set of vectors $b' |\xi_{\eta}^{\cB} \rangle$ is dense and hence can be used to approximate
$|\xi_{\psi}^{\cB} \rangle$ in norm. Then, since all operators in the inequality are bounded (in particular $V_\eta$, which is a contraction), 
this gives
\ben
\gamma_t(A^*A) \le \langle  \xi_{\psi}^{\cB} |  \alpha_{\eta,T}^t(A^*A)  \xi_{\psi}^{\cB} \rangle,
\een
which is equivalent to the desired bound.
\end{proof}

By this lemma, together with paper I  \cite{Faulkner:2020iou}, lem. 3 (1), and the monotonicity of the fidelity \cite{Uhlmann1977} 
we have $\| \Gamma_\psi(1/2+it)\|_{1,\psi} = F(\gamma_t|\omega_\psi) \le F(\omega_\psi \circ T \circ \alpha_{\eta,T}^t | \omega_\psi)$, with 
$L_1$-norm relative to $\cA'$ and $\psi$.  Thm.  \ref{thm:1} for the case $q=1$ therefore gives us
\ben
\label{eq:main1}
S(\omega_\psi | \omega_\eta) - S(\omega_\psi \circ T | \omega_\eta \circ T)
\ge 
-2 \int_{-\infty}^\infty p(t) \log F(\omega_\psi \circ T \circ \alpha_{\eta,T}^t|\omega_\psi)
\, \dd t . 
\een
Since ln and $F$ are concave, we can further use the Jensen inequality to pull the integral inside the fidelity, see 
paper I  \cite{Faulkner:2020iou}, lem. 10 and proof of thm. 2 for the identical details of these arguments. This completes the proof.
\end{proof}

\noindent
{\bf Remark:} The attentive reader will have noticed that thm. \ref{thm:1} can evidently 
be combined with lem. \ref{lem:mon} also for $q \in (1,2)$. 
Indeed, by \cite{Hollands2}, cor. 2
\ben
\omega_{\zeta_1}(aa^*) \ge \omega_{\zeta_2}(aa^*) \quad \forall a \in \cA \quad 
\Longrightarrow \quad \| \zeta_1 \|_{q,\psi,\cA'} \ge \| \zeta_2 \|_{q,\psi,\cA'}. 
\een
Therefore, lem. \ref{lem:mon} 
and thm. \ref{thm:1} with $q=2s$ give
\ben
\label{eq:Jensen}
S(\omega_\psi | \omega_\eta) - S(\omega_\psi \circ T | \omega_\eta \circ T)
\ge \frac{1-s}{s}  \int_{-\infty}^\infty p(t)  D_{s} ( \omega_\psi \circ T \circ \alpha_{\eta,T}^t | \omega_\psi)
\, \dd t .
\een
Here, $D_s, s \in (1/2,1)$ are the ``sandwiched Renyi divergences'' \cite{mueller}, defined by 
\ben
D_s(\omega_\zeta | \omega_\psi) = (s-1)^{-1} \, \log \| \zeta \|_{2s,\psi,\cA'}^{2s}
\een
 (norm relative to $\cA'$). 
We know  \cite{Berta2}, thm. 14 or \cite{Hollands2}, cor. 3 that these satisfy the 
data processing inequalitiy. Applying this inequality to the completely positive channel 
$ \cA \to \cA \oplus \cdots \oplus \cA, a \mapsto a \oplus \cdots \oplus a$ and the 
states $\rho = \oplus_i \lambda_i \omega_{\psi_i}, \sigma = \oplus_i \lambda_i \omega_{\zeta_i}$ implies that $D_s$ is 
jointly convex by a standard argument, see e.g. \cite{mueller}, proof of prop. 1, 
\ben
\sum_i \lambda_i D_s(\omega_{\zeta_i} | \omega_{\psi_i}) \ge D_s(\sum_i \lambda_i \omega_{\zeta_i} | \sum_j \lambda_j \omega_{\psi_j})
\een
where the sum is finite and $\lambda_i \ge 0, \sum \lambda_i=1$.
This suggests that we can use again the 
Jensen inequality to pull the integral inside $D_s$ in \eqref{eq:Jensen}. We would get, for $s \in (1/2,1)$
\ben
S(\omega_\psi | \omega_\eta) - S(\omega_\psi \circ T | \omega_\eta \circ T)
\ge \frac{1-s}{s} D_{s} (\omega_\psi \circ T \circ \alpha | \omega_\psi ). 
\een
The right side has an operational meaning in terms of hypothesis testing, see \cite{om}. We omit the details since \eqref{eq:main1} seems more fundamental anyhow.

\section{Examples}\label{sec:4}

\subsection{Relation with conditional expectation}

Let $\cB \subset \cA$ be an inclusion of v. Neumann factors.
Consider the reference state $\sigma$ fixed by a conditional expectation $E: \cA \to \cB$. 
We work out the estimate in Thm. \ref{thm:2} in this case. The rotated Petz map reads (see \cite{Faulkner:2020iou} and also \cite{AC}, \cite{Petz1993})
\ben
\alpha^t_\sigma = 
\iota \circ \varsigma_{\sigma; \cB}^{t} \circ j_{\cB} {\rm Ad}_{V^\star}  j_{\cA} \circ \varsigma_{\sigma; \cA}^{-t} 
\een
Here $j_\cA = {\rm Ad} J_\cA$ and $\varsigma_{\sigma,\cA}^t$ is the modular flow of $\sigma$ on $\cA$ etc.
In the above expression $\cB$ is represented on a different Hilbert space
since we do not always have a common cyclic and separating vector for $\mathcal{A},\mathcal{B}$. 
In particular in the case where $\sigma$ is fixed by a conditional expectation, we do not have such a vector.
$V$ is the isometry that relates the two Hilbert spaces
and is defined via $V(b \left| \eta_{\mathcal{B}} \right>) = \iota(b) \left| \eta_{\mathcal{A}} \right>$ (following \cite{Faulkner:2020iou} we distinguish these representations by writing explicitly the inclusion
$\iota(\mathcal{B}) \subset \mathcal{A}$.) If $\sigma$ is fixed by $E$ then we can dramatically simplify the recovery map 
using Takesaki's theorem \cite{Takesaki2} to cancel all the modular operators.
In particular one finds:
\ben
\alpha^t_\sigma = \iota \circ {\rm Ad}_{V^\star} = E
\een
See, for example, Eq~3.3-3.5 of \cite{Faulkner:2020hzi} for more details on this case (to compare to this paper send $\iota \rightarrow \beta$ and $\mathcal{A}, \mathcal{B} \rightarrow \mathcal{M} , \mathcal{N}$). 
The integral in the recovery channel  \eqref{eq:recov} is now trivial. It is well known that
the Petz map for an inclusion reduces to a conditional expectation if the defining state is fixed by a conditional expectation \cite{Petz1}. This remains true for the flowed Petz map. 
So the estimate becomes:
\begin{equation}
\label{est5}
S(\rho \, | \, \sigma \circ E) - S(\rho|_{\cB} \, | \, \sigma|_{\cB})  \geq - \log F( \rho \circ E, \rho)^2
\end{equation}
Notice however that this estimate can be derived in a more elementary manner. 
We have the well known equality (see \cite{Petz1993}):
\begin{equation}
S(\rho \, | \, \sigma \circ E) - S(\rho|_{\cB} \, | \, \sigma|_{\cB})
= S(\rho \, | \rho \circ E) 
\end{equation}
So \eqref{est5} simply follows from the monotonicity of the ``sandwiched Renyi divergence''
as a function of the Renyi parameter $s$ \cite{Berta2}. 

\subsection{Markov semi-groups}

Let $\{T_t\}_{t\ge 0}: \cA \to \cA$ be a weakly $^*$-continuous semi-group of 2-positive unital maps such that $T_0=id$, leaving invariant a reference state $\sigma$.  On physical 
grounds, further conditions are often imposed on such a ``Markov semi-group'' expressing a kind of {\it detailed balance 
condition}, see \cite{Koss,Maj1,Maj2}. One such condition that is natural from a physical viewpoint in many cases is that 
 \ben
 \sigma(a^* T_t(b))= \sigma(\Theta(b^*)T_t \circ \Theta(a))\quad (t \ge 0),
 \een 
where $\Theta$ is some anti-linear anti-automorphism\footnote{This means $\Theta(\lambda a)=\bar \lambda \Theta(a), 
\Theta(a)\Theta(b)=\Theta(ba)$ for $a,b \in \cA.$ It actually suffices that $\Theta$ is a Jordan map, see 
\cite{Maj} for details.} of $\cA$ leaving $\sigma$ invariant, 
such that $\Theta^2 = id$, often having the interpretation 
of a parity or PCT map. This can be seen to imply that $T_t$ commutes with the modular group $\varsigma^s_\eta$ of $(\cA,\sigma)$, 
see \cite{Maj}, lem. 4.8. This further gives (by analytic continuation and the KMS property)
\ben
\sigma(a^* \varsigma_\sigma^{-s+i/2} \circ T_t \circ \varsigma_\sigma^s (b)) = 
\sigma(a^*T_t \circ \varsigma_\sigma^{+i/2} (b)) = \sigma(\Theta \circ T_t \circ \Theta(a^*) \varsigma_\sigma^{+i/2} (b)), 
\een
showing in view of \eqref{eq:Petzdef} that the rotated Petz map of $T_t$ is given by $\Theta \circ T_t \circ \Theta$. 

Thus, in our case,  the recovery channel \eqref{eq:recov} for $T_t$ is (for given $t \ge 0$)
$
\alpha = \Theta \circ T_t \circ \Theta , 
$
whereas the formula for the entropy production by the Markov process up to time $t$ 
provided by thm. \ref{thm:2} in this case given in the following 
proposition:
\begin{proposition}
Let $\{T_t\}_{t\ge 0}: \cA \to \cA$ be a Markov semi-group leaving invariant a normal state $\sigma$ of the v. Neumann algebra $\cA$
fulfilling the above detailed balance condition.
Then for all $t \ge 0$,  
\ben
S(\rho | \sigma) - S(\rho \circ T_t | \sigma) \ge -2 \log F(\rho \circ T_{t} \circ \Theta \circ T_t \circ \Theta, \rho).
\een
\end{proposition}

\section{Proof of thm. \ref{thm:1}}\label{sec:5}

The proof is along similar lines as that of paper I  \cite{Faulkner:2020iou}, thm. 1 but  we take the opportunity to present alternative proofs in some cases
to obtain simplifications. The argument is divided into several steps. First we consider a special case involving an additional assumption. This is then removed in a second step.

We shall first assume  that there 
exists $c>0$ such that 
\ben
\label{eq:major}
\omega_\psi(a) \le c\omega_\eta(a)  \quad (a \in \cA^+).
\een
We will write this also as $ \omega_\psi \le c\omega_\eta $. By definition, since $T$ is positivity-preserving, we then also have 
$\omega_\psi \circ T \le c\omega_\eta \circ T$. Following the conventions introduced above around \eqref{eq:conventions}, we use notations such as 
$|\eta_\cB\rangle, |\eta_\cA\rangle$, in this proof, and similarly for $|\psi\rangle$. In the same spirit, 
we will sometimes write $\Delta_{\eta,\psi;\cB}=\Delta_{\eta_\cB,\psi_\cB}, \Delta_{\eta,\psi;\cA}=\Delta_{\eta_\cA,\psi_\cA}$ etc.

\medskip

Step 1): First, we analyze some properties of the family of operators appearing in \eqref{eq:Gupper}.  The following lemma 
does not use the majorization condition \eqref{eq:major} yet.

\begin{lemma}\label{lem:1}
The operator $\Delta_{\eta,\psi;\cA}^{z} V_\psi \Delta_{\eta,\psi;\cB}^{-z}$ satisfies 
 $\|\Delta_{\eta,\psi;\cA}^{z} V_\psi \Delta_{\eta,\psi;\cB}^{-z}\| \le 1$ in the closure of the strip $\bS_{1/2}=\{z \in \CC \mid 0 < \Re (z) < 1/2\}$. 
\end{lemma}

\begin{proof}
It follows from the Schwarz inequality \eqref{eq:eqschwarz} that $V_\psi^* \Delta_{\eta,\psi;\cA} V_\psi \le \Delta_{\eta,\psi;\cB}$, see \cite{Petz4}, proof of thm. 4. Let $\alpha \in [0,1]$. By the operator monotonicity 
of the function $\RR_+ \owns x \mapsto x^\alpha \in \RR_+$, it follows  $(V_\psi^* \Delta_{\eta,\psi;\cA} V_\psi)^\alpha \le \Delta_{\eta,\psi;\cB}^\alpha$, by L\"owner's theorem \cite{Hansen}. Furthermore, since 
$V_\psi$ is a contraction and since $\RR_+ \owns x \mapsto x^\alpha \in \RR_+$ is operator monotone, it follows from Jensen's operator inequality 
(see e.g. \cite{Petz1993}, lem. 1.2, or \cite{Petz4}, thm. C) 
$V_\psi^* \Delta_{\eta,\psi;\cA}^\alpha V_\psi \le (V_\psi^* \Delta_{\eta,\psi;\cA} V_\psi)^\alpha$, so $V_\psi^* \Delta_{\eta,\psi;\cA}^\alpha V_\psi \le \Delta_{\eta,\psi;\cB}^\alpha$. This is the same as saying 
\ben
\Delta_{\eta,\psi;\cB}^{-\alpha/2} V_\psi^* \Delta_{\eta,\psi;\cA}^\alpha V_\psi \Delta_{\eta,\psi;\cB}^{-\alpha/2} \le 1, 
\een
which is the same as saying $\|\Delta_{\eta,\psi;\cA}^{\alpha/2} V_\psi \Delta_{\eta,\psi;\cB}^{-\alpha/2} \| \le 1$. The statement follows because $\Delta_{\eta,\psi;\cA}^{it}, \Delta_{\eta,\psi;\cB}^{it}$ are isometric 
when $t$ is real. 
\end{proof}

Similarly as in paper I, we introduce a vector by
\ben\label{eq:G}
|\Gamma_\psi(z)\rangle \equiv   \Delta_{\eta_\cA,\psi_\cA}^{z} V_\psi  \Delta_{\eta_\cB,\psi_\cB}^{-z} |\xi_{\psi}^{\cB}\rangle.  
\een
This vector is well defined by the preceding lemma. In fact, it is analytic inside the strip $\bS_{1/2}$ and continuous 
on the boundaries of the strip, see e.g. paper I  \cite{Faulkner:2020iou}, thm. 4 (1), or the next lemma.  With the goal to extend the domain of analyticity of this vector, we let
$\Pi_\Lambda$ be the projector onto the spectral subspace of the modular operator $\log \Delta_{\psi_\cA}$, associated with the interval $(-\Lambda,\Lambda)$, for some large ``cutoff'' $\Lambda$ which 
we will let got to infinity. Then we define the new vector valued function $\Pi_\Lambda |\Gamma_\psi(z)\rangle$ on the one-sided strip $\bS_{1/2}=\{z: 0<\Re z<1/2\}$. 
The following two lemmas, which also rely on the majorization condition \eqref{eq:major}, give 
an analytic continuation to the two-sided strip $\{z: -1/2<\Re z<1/2\}$. They express the same 
as paper I, lem. 6 (1,2) but we give a somewhat different proof. 

\begin{lemma}\label{lem:5}
$\Pi_\Lambda |\Gamma_{\psi}(z)\rangle$ is well-defined and holomorphic in $\{z: -1/2<\Re z<1/2\}$ and strongly continuous at $\Re z=\pm 1/2$.
\end{lemma}
\begin{proof}
$(D\eta_\cB:D\psi_\cB)_{iz} |\xi_{\psi}^{\cB}\rangle = \Delta_{\eta_\cB,\psi_\cB}^{-z} |\xi_{\psi}^{\cB}\rangle$
is holomorphic in $-1/2<\Re z<0$ and strongly continuous at the boundary, 
as guaranteed by Tomita-Takesaki theory. The majorization condition 
in combination with lem. \ref{lem:2} shows that $(D\eta_\cB:D\psi_\cB)_{iz}$  is bounded 
for $1/2 \ge \Re z \ge 0$ so $ |\xi_{\psi}^{\cB}\rangle $ is in the domain of $\Delta_{\eta_\cB,\psi_\cB}^{-z}$ in this range. 
Then, \cite{CecchiniPetz}, lem. 2.1 shows that $\Delta_{\eta_\cB,\psi_\cB}^{-z} |\xi_{\psi}^{\cB}\rangle$ is holomorphic in this range, 
and strongly continuous at the boundary $\Re z=1/2$ and at $\Re z=0$.
Since $V_\psi$ is bounded, $V_\psi (D\eta_\cB:D\psi_\cB)_{iz} |\xi_{\psi}^{\cB}\rangle$
is holomorphic in $-1/2<\Re z<1/2$ by the edge-of-the-wedge theorem, and strongly continuous at $\Re z=1/2, -1/2$. 
Now we consider separately the following cases: 

(1)
$1/2 \ge \Re z \ge 0$: Then by lem. \ref{lem:2}, $(D\eta_\cB:D\psi_\cB)_{iz}  \in \cB$,
 and we may write 
$\Pi_\Lambda |\Gamma_{\psi}(z)\rangle
=\Pi_\Lambda \Delta_{\eta_\cA,\psi_\cA}^{z} T[(D\eta_\cB:D\psi_\cB)_{iz}] |\xi_{\psi}\rangle$
from the definition of $V_\psi$, and therefore $T[(D\eta_\cB:D\psi_\cB)_{iz}] |\xi_{\psi}\rangle$ is 
in the domain of $\Delta_{\eta,\psi}^{z} $ by standard Tomita-Takesaki theory. Thus, \cite{CecchiniPetz}, lem. 2.1 shows that
$\Pi_\Lambda |\Gamma_{\psi}(z)\rangle$ is holomorphic in the the $\bS_{1/2}$ and strongly continuous on its boundary.

(2) $-1/2\le \Re z\le 0$: Then by lem. \ref{lem:2}, $(D\eta_\cA:D\psi_\cA)_{iz}^{-1} = (D\eta_\cA:D\psi_\cA)_{-i\bar z}^* \in \cA$ is holomorphic 
for $z \in -\bS_{1/2}$ 
and strongly continuous on the boundaries. We may write 
\ben
\Pi_\Lambda |\Gamma_{\psi}(z)\rangle= \Pi_\Lambda \Delta_{\psi_\cA}^{z} (D\eta_\cA:D\psi_\cA)_{iz}^{-1} V_\psi 
\Delta_{\eta_\cB, \psi_\cB}^{-z} |\xi_{\psi}^{\cB}\rangle,
\een 
which is now seen to be holomorphic because 
$\Pi_\Lambda \Delta_{\psi_\cA}^{z}$ is a bounded holomorphic family of operators for $z \in -\bS_{1/2}$, strongly continuous on the boundaries, 
with norm not exceeding $e^{2\Lambda |\Re z|}$. 

Cases (1) and (2) together give the statement by an application of the operator-valued edge-of-the-wedge theorem. 
\end{proof}

\begin{lemma}\label{lem:2}
Suppose $|\phi\rangle, |\psi\rangle$ are states on a v. Neumann algebra $\cM$ in standard form. 
If there is $c>0$ such that $\omega_\zeta \le c \omega_\phi$, then $z \mapsto (D\phi:D\psi)_{iz}= \Delta_{\phi,\psi}^{-z} \Delta_{\psi}^{z}$
 is holomorphic for $z \in \bS_{1/2}$, strongly continuous at the boundary, and 
one has $\| (D\phi:D\psi)_{iz} \| \le c^{2\Re(z)}$, $(D\phi:D\psi)_{iz} \in \cA$.
\end{lemma}

\begin{proof}
This fact is well known and we include the proof only for completeness. 
The following calculation shows that $\Delta_{\psi} \le c \Delta_{\phi,\psi}$:
\ben
\langle a \xi_\psi | \Delta_{\phi,\psi} a \xi_\psi \rangle = \langle \xi_\phi | a \pi^\cM(\psi) a^* \xi_\phi \rangle \ge c^{-1} \langle \xi_\psi | a \pi^{\cM'}(\psi) a^* \xi_\psi \rangle 
=c^{-1} \langle a \xi_\psi | \Delta_\psi a \xi_\psi \rangle
\een
where $a \in \cA$, which uses $\omega_\psi \le c\omega_\phi$. Then, by operator monotonicity  \cite{Hansen}, we get $\Delta_\psi^\alpha \le c^\alpha \Delta_{\phi,\psi}^\alpha$
for $0 \le \alpha \le 1$, and thereby $\| \Delta_{\phi,\psi}^{-z} \Delta_{\psi}^{z} \| \le c^{2\Re(z)}$ when $z$ is in the closure of $\bS_{1/2}$. If $a \in \cM$, 
standard Tomita-Takesaki theory tells us that $\Delta_{\psi}^{z} a |\psi\rangle$ is holomorphic inside $\bS_{1/2}$ and strongly continuous
at the boundary. Therefore, these vectors are in the domain of  $\Delta_{\phi,\psi}^{-z}$ and then \cite{CecchiniPetz}, lem. 2.1 tells 
us that $\Delta_{\phi,\psi}^{-z} \Delta_{\psi}^{z} a |\psi\rangle$ is holomorphic in $\bS_{1/2}$ and strongly continuous at the boundary, so 
$(D\phi:D\psi)_{iz}$ is a holomorphic operator valued function on $\bS_{1/2}$ which is strongly continuous at the boundary. 
For $z \in i\RR$, 
we know in general that $(D\phi:D\psi)_{iz} \in \cM$, so for any $a' \in \cM'$, $z \mapsto [a', (D\phi:D\psi)_{iz}]$ is a holomorphic 
function in $\bS_{1/2}$ taking values in the bounded operators with a bounded, vanishing limit as $\Re z \to 0^+$. Then, by the by the operator-valued 
edge of the wedge theorem, the function vanishes identically and we get $(D\phi:D\psi)_{iz} \in \cM''=\cM$ when $z$ is in the closure of $\bS_{1/2}$. 
\end{proof}

Step 2)
Let $1/p_\theta=(1-2\theta)/2+2\theta/q$ for $\theta \in [0,1/2]$, which interpolates between $p_0=2,p_{1/2}=q \in [1,2]$. Very similarly as in paper I, the relationship between the function \eqref{eq:G} and the 
quantum data-processing inequality is that  
\ben
\label{eq:limit}
\lim_{\theta \to 0^+} \left( -\frac{1}{\theta} \log \| \Pi_\Lambda \Gamma_\psi(\theta) \|_{p_\theta,\psi}^2 \right) = S(\psi_\cA | \eta_\cA) - S(\psi_\cB | \eta_\cB). 
\een
On the left side, we mean the Araki-Masuda $L_{p_\theta}(\cA',\psi)$-norm, i.e. relative to the commutant algebra $\cA'$ and the state $|\psi\rangle$.
For finite dimensional quantum systems, this is a straightforward consequence of the definitions. In the general v. Neumann context, we could proceed 
in basically the same way as in paper I, but here we can take certain shortcuts since we are not interested in reproving paper I, thm. 3. 

Setting $|\zeta_\theta\rangle = \Pi_\Lambda |\Gamma_\psi(\theta)\rangle/\| \Pi_\Lambda \Gamma_\psi(\theta)\|$, the homogeneity of the $L_p$-norm gives
\ben
\label{eq:AB}
\begin{split}
\lim_{\theta \to 0^+} \left( -\frac{1}{\theta} \log \| \Pi_\Lambda \Gamma_\psi(\theta) \|_{p_\theta,\psi}^2 \right) &= \lim_{\theta \to 0^+} \left( -\frac{1}{\theta} 
\log \| \Pi_\Lambda \Gamma_\psi(\theta) \|^2 \right) + \lim_{\theta \to 0^+} \left( -\frac{1}{\theta} \log  \| \zeta_\theta \|_{p_\theta,\psi}^2
 \right) \\
&=: A+B.
\end{split}
\een
Note that $\lim_{\theta \to 0} \| \zeta_\theta - \psi\|^2/\theta = 0$ by construction as $\Pi_\Lambda|\Gamma_\psi(z)\rangle$ is 
strongly holomorphic at $z=0$ and $\Pi_\Lambda|\Gamma_\psi(0)\rangle=|\psi\rangle$ by definition. 
The term $B$ is treated with the next two lemmas playing the same role as paper I, lem. 4 and thm. 3 (2). 

\begin{lemma}\label{lem:squeeze}
For two unit norm vectors $|\zeta\rangle,|\psi\rangle \in {\mathscr Hs}$ and a v. Neumann algebra $\cM$ acting on $\mathscr H$, 
$p \in [1,2]$, we have 
\ben
| \langle \zeta | \psi \rangle | \le \| \zeta \|_{p,\psi} \le \| \Delta_{\psi,\zeta}^{1/p-1/2} \zeta \|.
\een
\end{lemma}

\begin{proof}
The statement holds trivially for $p=2$, because the $L_2$-norm is the Hilbert-space norm on vectors in $\mathscr H$. For $p=1$, the right 
inequality holds by the ``inf'' variational definition of the $L_1$-norm, see \eqref{eq:pnorm}, and the left inequality by the 
``sup'' variational definition of the fidelity \eqref{eq:Fchar}, which is equal to the $L_1$-norm, as proven 
in the generality needed here in paper I, lem. 3 (1). 
So we can assume $p \in (1,2)$ in the rest of the proof.

(1) The upper  bound is trivial by the ``inf'' definition of the Araki-Masuda $L_p$-norm (valid when $p \in [1,2]$):
\ben
\| \zeta\|_{p,\psi}^2  = \inf_{\|\phi\|=1} \| \Delta_{\phi,\psi}^{1/2-1/p} \zeta \|^2 \le \| \Delta_{\zeta,\psi}^{1/p-1/2} \zeta \|^2 
=  \| \Delta^{1/p-1/2}_{\psi,\zeta} \zeta \|^2.
\een
(2) For a faithful $|\psi\rangle$, the lower bound follows immediately from the H\" older inequality proven by \cite{AM}:
$| \langle \zeta | \psi \rangle | \le \| \zeta \|_{p,\psi} \| \psi \|_{p',\psi}$ with $p'$ the dual H\" older index of $p$.
 The proof is then completed by $\| \psi \|_{p',\psi} =1$, also shown by \cite{AM}. For non-faithful $|\psi\rangle$
 this proof would also go through if one could generalize the H\" older inequality to this situation which appears likely. 
 Here we give a simpler argument 
 based on the variational characterization \cite{Hollands2}, prop. 1 of the $L_p$ norms. 
 First, from the supremum characterization of the fidelity \eqref{eq:Fchar}, we 
 have $F(y\omega_\zeta y^*,\omega_\psi)^2\ge |\langle \psi | y \zeta \rangle|^2 = \| P_\psi y \zeta\|^2$, with $P_\psi = |\psi \rangle \langle \psi |$.
 Then, using the notations of \cite{Hollands2}, prop. 1 
 \ben
 \begin{split}
\| \zeta \|_{p,\psi,\cM}^2
\ge &
-\frac{\sin(2\pi/p)}{\pi}
\inf_{x: \RR_+ \to \cM'} \int_0^\infty [\|x(t) \zeta \|^2  + t^{-1} \| P_\psi y(t) \zeta\|^2] t^{-2(1-1/p)} \dd t \\
= &
-\frac{\sin(2\pi/p)}{\pi} \int_0^\infty \langle \zeta | P_\psi(t+P_\psi)^{-1}  \zeta\rangle t^{-2(1-1/p)} \dd t \\
= &
-\| P_\psi \zeta\|^2 \frac{\sin(2\pi/p)}{\pi} \int_0^\infty (t+1)^{-1}  t^{-2(1-1/p)} \dd t 
=  | \langle \zeta | \psi \rangle |^2.
\end{split}
\een
The lemma is proven.
\end{proof}

\begin{lemma}
\label{lem:firstlaw}
For a v. Neumann algebra $\cM$ acting on $\sH$, if $|\psi\rangle \in \sH$ is a unit vector 
and $|\zeta_\theta \rangle$ is a 1-parameter family of unit vectors such that 
\ben
\lim_{\theta \to 0} \frac{\| \zeta_\theta - \psi\|^2}{\theta} = 0,
\een 
then we have 
\ben
\lim_{\theta \to 0} \frac{1}{\theta} \log \|\zeta_\theta \|_{p_\theta,\psi} = 0
\een
for any family $p_\theta$ such that $1 \le p_\theta \le 2$.
\end{lemma}
\begin{proof}
By the Hadamard three lines theorem, $\| \Delta_{\psi,\zeta}^{1/p-1/2} \zeta \| \le C_0^{1/p-1/2} C_1^{1-1/p}$, where 
$C_0 = \sup_t \| \Delta^{it}_{\psi,\zeta} \zeta \| = 1, C_1 = \sup_t \| \Delta^{1/2 + it}_{\psi,\zeta} \zeta \| = 1$, so 
$\| \Delta_{\psi,\zeta}^{1/p-1/2} \zeta \| \le 1$. In view of 
${\rm Re} \langle \zeta | \psi \rangle \le | \langle \zeta | \psi \rangle |$, the inequalities of lem. \ref{lem:squeeze} give
\ben
\frac{1}{2} \| \psi-\zeta\|^2 = 1- {\rm Re} \langle \zeta | \psi \rangle  \ge 1- \| \zeta \|_{p,\psi} \ge 0.
\een
The chain rule for ln then proves the lemma if we substitute $|\zeta\rangle=|\zeta_\theta\rangle$, divide by $\theta$
and pass to the limit. 
\end{proof}

Lem. \ref{lem:firstlaw} applied to $|\zeta_\theta\rangle$ and $\cM=\cA'$ implies that $B=0$ in \eqref{eq:AB}. Now we study the term $A$ in that equation. By the chain rule
\ben\label{eq:calc1}
\begin{split}
 -\frac{\dd}{\dd \theta} \| \Pi_\Lambda \Gamma_\psi (\theta) \|^2 \bigg|_{\theta=0} &=  -2\frac{\dd}{\dd \theta} \Re \langle \xi_{\psi}^{\cA} |\Pi_\Lambda \Gamma_\psi(\theta) \rangle \bigg|_{\theta=0} \\
 &=  -2\frac{\dd}{\dd \theta} \langle \xi_{\psi}^{\cA} |  \Pi_\Lambda \Delta_{\psi_\cA}^{\theta} (D\eta_\cA:D\psi_\cA)_{i\theta}^{-1} V_\psi (D\eta_\cB:D\psi_\cB)_{i\theta} \xi_{\psi}^{\cB} \rangle \bigg|_{\theta=0}. 
\end{split} 
\een
Here we use the Connes-cocycle $(D\psi:D\phi)_t = \Delta_{\psi,\phi}^{it} \Delta_{\phi,\phi}^{-it}$, or rather its analytic continuation to imaginary $t$-values. The derivative can be also calculated by analytic continuation $\theta \to i\theta$ by lem. \ref{lem:5}, and after that, the Connes-cocycles 
$(D\eta_\cA:D\psi_\cA)_{\theta}^{-1}, (D\eta_\cB:D\psi_\cB)_{-\theta}$ (real argument) are elements of $\cA, \cB$ respectively, by standard results of relative modular theory. 
Then we are allowed to use 
the definition of $V_\psi$, eq. \eqref{eq:Vdef}, and $ \Delta_{\psi_\cA}^{\theta}|\xi_{\psi}^{\cA} \rangle = |\xi_{\psi}^{\cA} \rangle$. We obtain:
\ben\label{eq:calc2}
 -\frac{\dd}{\dd \theta} \| \Pi_\Lambda \Gamma_\psi(\theta) \|^2 \bigg|_{\theta=0}
 = -2i\frac{\dd}{\dd \theta} \langle \xi_{\psi}^{\cA} | (D\eta_\cA:D\psi_\cA)_{\theta}^{-1} T[(D\eta_\cB:D\psi_\cB)_{\theta}] \xi_{\psi}^{\cA} \rangle \bigg|_{\theta=0}.
\een
Finally, we distribute the derivatives on the right side and apply the definition of the relative entropy in terms of the Connes-cocycle, showing in view of 
\eqref{eq:AB} and $B=0$ that \eqref{eq:limit} holds. 

\medskip

Step 3): In view of step 2), we should next look for an upper bound on $ \|\Pi_\Lambda \Gamma_\psi(\theta) \|_{p_\theta,\psi}$
(with $L_p$-norm relative to $\cA'$ and $|\psi\rangle$).  
We apply the following lemma (paper I, lem. 9) to $|G(z)\rangle = \Pi_\Lambda |\Gamma_\psi(z)\rangle$, $p_0=2, p_{1/2}=q$ and 
$\cM=\cA'$:

\begin{lemma}
\label{lem:hirsch}
Let $|G(z)\rangle$ be a $\sH$-valued holomorphic function on the strip $\bS_{1/2}=\{0<{\rm Re}z<1/2\}$ 
that is uniformly bounded in the closure, $|\psi\rangle \in \sH$ a possibly non-faithful state of a
sigma-finite v. Neumann algebra $\mathcal M$ in standard form acting on $\sH$. Then, for $0<\theta<1/2$, 
\ben
\frac{1}{p_\theta} = \frac{1-2\theta}{p_0} + \frac{2\theta}{p_1}
\een
with $p_0,p_{1/2} \in [1,2]$, we have 
\begin{align}
\label{himp}
& \log \left\| G(\theta)\right\|_{p_\theta, \psi, \cM} \\
  \leq  & \int_{-\infty}^{\infty} {\rm d} t 
  \left(
 (1-2 \theta)  \alpha_\theta(t) \log \left\| G(it) \right\|_{p_0, \psi,\cM} + (2\theta)  \beta_\theta(t) 
  \log  \left\| G(1/2+it) \right\|_{p_{1/2}, \psi,\cM} \right),
  \nonumber
\end{align}
where
\begin{equation}
\alpha_\theta(t) = \frac{ \sin(2\pi\theta)}{(1-2\theta)(\cosh(2\pi t ) - \cos(2\pi \theta)) }\,,
\qquad \beta_\theta(t) = \frac{ \sin(2\pi\theta)}{2 \theta(\cosh(2\pi t ) + \cos(2\pi \theta)) }.
\end{equation}
\end{lemma}
\noindent
We furthermore use that 
\ben
\| \Pi_\Lambda \Gamma_\psi(it) \|_{2,\psi,\cA'} = \|\Pi_\Lambda \Gamma_\psi(it) \| 
= \|  \Delta_{\eta,\psi;\cA}^{it} V_\psi \Delta_{\eta,\psi;\cB}^{-it} \xi_{\psi}^{\cB}\| \le \| V_\psi \| \le 1, 
\een
since $V_\psi$ is a contraction. Then we find using \eqref{eq:limit} that 
\ben
\label{eq:statL}
\begin{split}
& S(\psi_\cA | \eta_\cA) - S(\psi_\cB | \eta_\cB) \\
\ge &
-2\pi \int_{-\infty}^\infty \left[ 1+\cosh(2\pi t) \right]^{-1} 
 \log \| \Pi_\Lambda \Delta_{\eta,\psi;\cA}^{1/2+it} V_\psi \Delta_{\eta,\psi;\cB}^{-1/2-it} \xi_{\psi}^{\cB} \|_{q,\psi,\cA'} \dd t. 
 \end{split}
\een  
We can let $\Lambda \to \infty$, using the 
continuity of the $L_q$-norm in the norm-topology on the Hilbert space $\sH$, 
for $1 \le q \le 2$, \cite{AM}, lem. 6.1 (2), thus allowing us to replace $\Pi_\Lambda$ with $1$.
This completes the proof of the theorem under the majorization hypothesis \eqref{eq:major}.

Step 4):
We shall now remove the majorization hypothesis \eqref{eq:major} by the same method as in paper I. 
This hypothesis was crucially used only in lem. \ref{lem:2}, which in turn was used in order to show that the regularized 
vector $\Pi_\Lambda|\Gamma_\psi(z)\rangle$ is holomorphic near $z=0$. Since we now no longer want to assume 
the majorization hypothesis, we shall introduce a ``Gaussian'' regularization. For $P>0$, we consider the scaled Gaussians 
\ben
\hat g_P(x) = e^{-x^2/(2P)}  \quad \Longrightarrow \quad g_P(k) = \sqrt{P/\pi} e^{-Pk^2/2} . 
\een
Then we define
\ben
\label{eq:psiPdef}
|\psi_P\rangle := \hat g_P(\log \Delta_{\eta,\psi}) |\xi_\psi \rangle/\|  \hat g_P(\log \Delta_{\eta,\psi}) |\xi_\psi \rangle \|. 
\een
It follows from the spectral theorem that the corresponding functional $\omega_{\psi_P} \to \omega_\psi$ in norm. 
Thus, the state vector representatives in the natural cone converge strongly in view of the well-known inequality
$\|\omega_{\psi_P} -  \omega_\psi\| \ge \| \xi_\psi -\xi_{\psi_P}\|^2$.
In lem. \ref{lem:3}, 3) we show that $\omega_{\psi_P} \le c_P \omega_\eta$, so the analysis of the previous steps is applicable and the 
theorem holds for $\omega_{\psi_P}$ in place of $\omega_\psi$. 

Prop. \ref{prop:1} shows   in particular that 
$
\limsup_{P \to \infty} S(\psi_{P\cA} | \eta_{\cA}) = S( \psi_{\cA} | \eta_{\cA}), 
$
whereas lower semi-continuity of the relative entropy \cite{Araki2} gives us
$
S(\psi_\cB | \eta_\cB) \le \liminf_{P \to \infty} S(\psi_{P\cB} | \eta_{\cB}).
$
Thus it remains to be shown that the lower bound in the statement of the theorem converges for $P \to \infty$.

\cite{Araki2}, lem. 4.1 shows  that $\Delta_{\eta,\psi_P;\cA} \to \Delta_{\eta,\psi;\cA}$ in the 
strong resolvent sense. Therefore, for any bounded continuous function $f: \RR_+ \to \CC$, $f(\Delta_{\psi_P,\eta;\cA}) \to f(\Delta_{\psi,\eta;\cA})$ strongly, \cite{specth}, prop. 10.1.9, and therefore, for example, $(i+\log \Delta_{\psi_P,\eta; \cA})^{-1} \to (i+\log \Delta_{\psi,\eta; \cA})^{-1}$
in the strong sense. By \cite{specth}, prop. 10.1.8, $\Delta_{\psi_P,\eta; \cA}^{it} \to \Delta_{\eta,\psi;\cA}^{it}$
strongly. Then, $|\Gamma_{\psi_P}(1/2+it)\rangle \to |\Gamma_{\psi}(1/2+it)\rangle$ strongly, and therefore 
$\| \Gamma_{\psi_P}(1/2+it)\|_{1,\psi_P} \to \| \Gamma_{\psi}(1/2+it)\|_{1,\psi}$ by \cite{AM}, lem. 6.1 (2) for fixed $t$ or by the strong continuity of the fidelity, see paper I, lem. 11. Therefore,
$\| \Gamma_{\psi_P}(1/2+it)\|_{q,\psi_P} \to \| \Gamma_{\psi}(1/2+it)\|_{q,\psi}$ for $q \in [1,2]$ by \cite{Hollands2}, prop. 1 since the 1-norm is the 
fidelity. The lower bound in \eqref{eq:statL} thereby converges to the lower bound claimed in the theorem by dominated convergence.

\vspace{1cm}

{\bf Acknowledgements:} SH\ is grateful to the Max-Planck Society for supporting the collaboration between MPI-MiS and Leipzig U., grant Proj.~Bez.\ M.FE.A.MATN0003. He thanks Felix Otto for discussions and R. Longo for suggesting a connection with the index.
TF and SH benefited from the KITP program ``Gravitational Holography'' which was supported in part by the National Science Foundation under Grant No. NSF PHY-1748958.
TF acknowledges part of the work presented here is supported by the DOE under grant DE-SC0019517 as well as the Air Force Office of Scientific Research under award number FA9550-19-1-0360.

\appendix

\section{Continuity of relative entropy under Gaussian regularization}

The following proposition is strictly speaking a special case of paper I, thm. 5 but we give a somewhat different proof. 

\begin{proposition}\label{prop:1}
Let $\omega_\psi, \omega_\eta$ be normal state functionals on $\cM$, a v. Neumann algebra in standard form acting on $\sH$, and let $\omega_{\psi_P}$ be the regularized state functionals defined by the Gaussian regularization of the vectors as in eq. \eqref{eq:psiPdef}. Then $\lim_{P \to \infty} S(\psi_{P} | \eta) = S( \psi | \eta)$.
\end{proposition}
\begin{proof}
We have
$
S(\psi | \eta) \le \liminf_{P \to \infty} S(\psi_{P} | \eta)
$
by lower semi-continuity of the relative entropy, see \cite{Araki2}. Therefore, we only need to show that 
$\limsup_{P \to \infty} S(\psi_{P} | \eta) \le S( \psi | \eta)$. 

The following lemma gives the key properties of $\psi_P$:
\begin{lemma}\label{lem:3}
\begin{enumerate}
\item There is $a_P \in \cM$ such that $|\psi_P\rangle = a_P |\xi_\psi\rangle$ and
$\| a_P\| \le 1$ for all $P$.
\item For $0 \le \alpha \le 1$, we have 
\ben
\Delta_{\psi_P,\eta}^{\alpha} \ge \| a_P \xi_\psi \|^{-2\alpha}  a_P \Delta_{\psi,\eta}^{\alpha} a_P^*
\een
\item For some $c_P<\infty$, we have $\omega_{\psi_P} \le c_P\omega_\eta$. 
\end{enumerate}
\end{lemma}
\begin{proof}
1) This follows from paper I, lem. 5 (1). 

2) An easy calculation gives that $\Delta_{\psi_P,\eta} = \| a_P \xi_\psi \|^{-2}  a_P \Delta_{\psi,\eta} a_P^*$. 
Then, since $\|a_P\|\le 1$ by 1), it is legitimate to use Jensen's inequality (see e.g. \cite{Petz1993}, lem. 1.2, or \cite{Petz4}, thm. C)  
applied to the operator monotone function $x^\alpha$. This gives the desired inequality.

3) This follows from paper I, lem. 5 (4,5).
\end{proof}

We continue with the following representation of the relative entropy:
\ben
S(\psi_P | \eta) = \lim_{\alpha \to 1^-} \frac{1}{\alpha-1} \log \langle \xi_\eta | \Delta^{\alpha}_{\psi_P,\eta} \xi_\eta \rangle. 
\een 
Now we use lem. \ref{lem:3}, 2), and $S_{\psi,\eta} = J \Delta_{\psi,\eta}^{1/2}$:
\ben
\begin{split}
S(\psi_P | \eta) 
\le& \lim_{\alpha \to 1^-} \frac{1}{\alpha-1} \log \frac{  \langle a_P^* \xi_\eta | \Delta^{\alpha}_{\psi,\eta} a_P^* \xi_\eta \rangle}{\|a_{P} \xi_\psi\|^2} 
-2 \log \|a_P \xi_\psi \| \\ 
=& \lim_{\alpha \to 1^-} \frac{1}{\alpha-1} \log \frac{  \langle a_P \xi_\psi | J \Delta^{\alpha-1}_{\psi,\eta} J a_P \xi_\psi \rangle}{\|a_{P} \xi_\psi\|^2} 
-2 \log \|a_P \xi_\psi \| \\ 
=& \lim_{\alpha \to 1^-} \frac{1}{\alpha-1} \log \frac{  \langle a_P \xi_\psi | \Delta^{1-\alpha}_{\eta,\psi}  a_P \xi_\psi \rangle}{\|a_{P} \xi_\psi\|^2} 
-2 \log \|a_P \xi_\psi \| . 
\end{split}
\een
We write out again the definition of $a_P$ in this expression, paper I, proof of lem. 5 (1), obtaining
\ben
\begin{split}
S(\psi_P | \eta) \le&-\lim_{\theta \to 0^+} \frac{1}{\theta} \log 
\frac{  \langle \hat g_P(\Delta_{\eta,\psi}) \xi_{\psi} | \Delta_{\eta,\psi}^{\theta} \hat g_P(\Delta_{\eta,\psi}) \xi_{\psi} \rangle}{\| \hat g_P(\Delta_{\eta,\psi}) \xi_{\psi}\|^2} 
-2 \log \|a_P \xi_\psi \| \\
=& -\langle \hat g_P(\Delta_{\eta,\psi}) \xi_{\psi} | (\log \Delta_{\eta,\psi}) \hat g_P(\Delta_{\eta,\psi}) \xi_{\psi} \rangle -2 \log \|a_P \xi_\psi \| . 
\end{split}
\een
The term in the second line can be written in terms of the spectral measure $E_{\eta,\psi}(k) \dd k$ of $\log \Delta_{\eta,\psi}$ and then we can argue just as in paper I, thm. 5 (2) to show that
$\limsup_{P \to \infty} S(\psi_P | \eta) \le S(\psi | \eta)$. 
\end{proof}

\section{Weighted $L_p$ spaces \cite{AM}} 

The weighted $L_p$-spaces were defined by \cite{AM} relative to a fixed cyclic and separating vector $|\psi\rangle \in \H$ in the a natural cone of a standard representation of a v. Neumann algebra $\cM$. 
For $p \ge 2$, the space $L_p(\cM, \psi)$ is defined as 
\ben
L_p(\cM, \psi) = \{ |\zeta\rangle \in \bigcap_{|\phi\rangle \in \H} \sD(\Delta_{\phi, \psi}^{(1/2)-(1/p)}), \|\zeta\|_{p,\psi} < \infty\}.
\een
Here, the norm is 
\ben
\|\zeta\|_{p,\psi} = \sup_{\|\phi\|=1} \|\Delta_{\phi, \psi}^{(1/2)-(1/p)}\zeta\|.
\een
For $1 \le p < 2$, $L_p(\cM, \psi)$ is defined as the completion of $\sH$ with respect to the following norm:
\ben\label{eq:pnorm}
\|\zeta\|_{p,\psi} = \inf \{  \|\Delta_{\phi, \psi}^{(1/2)-(1/p)}\zeta\| : \|\phi\|=1, \pi^\cM(\phi) \ge \pi^\cM(\psi)=1 \}.
\een
 The generalization to non-faithful 
$\psi$, whose representing vector is not separating, is laid out some detail in \cite{Berta2}, or in \cite{Hollands2} using 
a variational method.

\end{document}